\documentclass[11pt,a4paper]{article}
\usepackage{amsfonts,amssymb}
\usepackage{amsopn}
\usepackage{amsmath,amsthm}
\usepackage{verbatim}

\DeclareMathOperator{\codim}{\textup{codim}}

\DeclareMathOperator{\R}{\mathbb{R}}

\DeclareMathOperator{\SU}{\textup{SU}}
\DeclareMathOperator{\sgn}{\textup{sgn}}

\DeclareMathAlphabet{\mathpzc}{OT1}{pzc}{m}{it}

\newcommand{\ba}{\boldsymbol{a}}
\newcommand{\bX}{\boldsymbol{X}}
\newcommand{\bx}{\boldsymbol{x}}
\newcommand{\by}{\boldsymbol{y}}
\newcommand{\liprod}{\ \raisebox{.2ex}{$\llcorner$}\ }
\newcommand{\riprod}{\ \raisebox{.2ex}{$\lrcorner$}\ }
\newcommand{\pmp}{\substack{+\\[-2pt](-)}}

\theoremstyle{plain}
\newtheorem{thm}{Theorem}

\newtheorem{lem}[thm]{Lemma}
\newtheorem*{cor}{Corollary}
\newtheorem{prop}[thm]{Proposition}

\theoremstyle{definition}

\theoremstyle{remark}

\newtheorem*{ack}{Acknowledgements}

\title{Geometric extensions of many-particle Hardy inequalities}

\author{Douglas Lundholm\thanks{e-mail: dogge@math.kth.se
	\hspace{3.0cm} MSC classes: 35R45, 81Q10, 81Q80}}

\date{\scriptsize{Department of Mathematics, KTH Royal Institute of Technology\\ SE-100 44 Stockholm, Sweden}}

\begin{document}

\maketitle

\begin{abstract}
	Certain many-particle Hardy inequalities 
	are derived in a simple and systematic way
	using the so-called ground state representation for the
	Laplacian on a subdomain of $\mathbb{R}^n$.
	This includes geometric extensions of the standard Hardy 
	inequalities to involve volumes of simplices spanned by a subset 
	of points. Clifford/multilinear algebra is employed to 
	simplify geometric computations.
	These results and the techniques involved are relevant for
	classes of exactly solvable quantum systems 
	such as the Calogero-Sutherland models
	and their higher-dimensional generalizations,
	as well as for membrane matrix models,
	and models of more complicated particle interactions 
	of geometric character.
\end{abstract}

\setlength\arraycolsep{2pt}
\def\arraystretch{1.2}

\section{Introduction}

	During the past century, Hardy inequalities have 
	appeared in a variety of different forms in the literature 
	and have played an important role in analysis and mathematical physics 
	(see e.g. the books \cite{Mazya:85,Kufner-Opic:90} 
	and the reviews in \cite{Davies:99,Tidblom:thesis}).
	The standard Hardy inequality associated to 
	the Laplacian in $\R^d$, $d \ge 3$,
	is given by
	\begin{equation} \label{standard-Hardy-ineq}
		\int_{\R^d} |\nabla u|^2 \,dx \ge C_d \int_{\R^d} \frac{|u|^2}{|x|^2} \,dx,
	\end{equation}
	with the sharp constant
	$C_d = \frac{(d-2)^2}{4}$, 
	where $u$ is any function in the Sobolev space $H^1(\R^d)$,
	i.e. for which the l.h.s. is finite.
	It states explicitly that the Laplace operator $-\Delta$ on $\R^d$
	(defined via its quadratic form)
	is not only non-negative,
	but in a sense strictly positive, 
	in that it is bounded below by a potential
	which even increases in strength unboundedly 
	as the distance to the origin tends to zero.
	In quantum mechanics,
	this is a concrete manifestation of the uncertainty principle,
	and inequalities of this form have been crucial for
	e.g. rigorous proofs for the stability of 
	matter (see e.g. \cite{Lieb-Seiringer:10}).
	In such many-particle contexts it becomes relevant to also consider 
	extensions 
	of \eqref{standard-Hardy-ineq} involving mutual
	distances between a (possibly large) number, say $N$, of particles.
	Also the sharp values of the corresponding constants 
	as well as their dependence on $N$
	are relevant for physical applications.
	
	An area where the importance of the Hardy inequality and its many-particle 
	generalizations becomes particularly transparent
	is in the context of 
	exactly solvable quantum systems.
	For example, in the Calogero-Sutherland model \cite{Calogero:69,Sutherland:71} 
	for $N$ particles on the real line $\R$ with inverse-square interactions, 
	the corresponding many-particle Hardy inequality 
	\begin{equation} \label{standard-1D-many-particle}
		\int_{\R^N} \sum_{j=1}^N |\partial_j u|^2 \,dx 
			\ge \frac{1}{2} \int_{\R^N} \sum_{j<k} \frac{|u|^2}{(x_j-x_k)^2} \,dx,
		\ \ u \in H^1_0\big( \R^N \cap \{x_j \neq x_k\}_{j\neq k} \big),
	\end{equation}
	guarantees the
	stability of the model, i.e. the boundedness of the 
	energy from below,
	for a certain range of coupling parameters and for a wide class
	of external potentials.
	Although first arising as toy models,
	the Calogero-Sutherland models have proved to be very useful 
	in the study of a variety of physical phenomena, 
	such as 
	soliton wave propagation \cite{Chen-Lee-Pereira:79},
	quantum spin chains \cite{Haldane:88,Shastry:88},
	random matrices \cite{Simons-Lee-Altshuler:94},
	as well as for 
	anyons in the lowest Landau level \cite{Hansson-Leinaas-Myrheim:92,Ouvry:07}.
	Various generalizations of the particle interactions 
	in the Calogero-Sutherland 
	model, including to higher dimensions and more complicated 
	geometric potentials, 
	have also been considered;
	see e.g. \cite{Calogero-Marchioro:73,Murthy-Bhaduri-Sen:96,Ghosh:97,Ghosh-Rao:98,Calogero:99}.
	Considering the success of the original Calogero-Sutherland models,
	some of these generalizations are likewise also expected to be important,
	for example in the study of strongly correlated systems
	(further physical motivations for studying such models are discussed 
	in \cite{Calogero-Marchioro:73,Murthy-Bhaduri-Sen:96,Ghosh:97,Ghosh-Rao:98,Khare-Ray:97}).

	Hoffmann-Ostenhof et. al. have in \cite{H-O2-Laptev-Tidblom:08}
	studied many-particle generalizations 
	of the standard Hardy inequality \eqref{standard-Hardy-ineq}
	of the conventional type \eqref{standard-1D-many-particle} 
	in arbitrary dimensions,
	both for bosonic and fermionic 
	particles (i.e. completely symmetric resp. antisymmetric wave functions),
	as well as for magnetically interacting particles
	in two dimensions, 
	and determined the optimal behavior for the associated constants
	in these inequalities in the large-$N$ limit 
	for the bosonic or dist\-inguish\-able case
	(hence with implications for the coupling parameters in  
	associated exactly solvable models).
	In the fermionic case, the optimal large-$N$ behavior of the
	corresponding constants was studied in \cite{Frank-H-O-Laptev-Solovej}.
	The magnetic case, 
	relevant for models of anyons in two dimensions, 
	was reconsidered and improved in \cite{Lundholm-Solovej:anyon},
	leading to rigorous bounds for the energy of the anyon gas
	(see also \cite{Lundholm-Solovej:extended,Lundholm-Solovej:exclusion} 
	for recent applications).
	
	In this note we will focus on the case of bosons or distinguishable particles
	and use the so-called
	ground state representation for the Laplacian on a subdomain
	of $\R^n$ to derive the conventional (bosonic) many-particle 
	Hardy inequalities in a simple and systematic way.
	Using multilinear and Clifford algebra,
	the approach we take straightforwardly generalizes to other
	types of many-particle Hardy inequalities,
	involving geometric relatives and higher-dimensional
	analogs of distances between particles.
	Some of these generalizations coincide with the interaction potentials 
	of generalized Calogero-Sutherland models such as those studied in 
	\cite{Calogero-Marchioro:73,Murthy-Bhaduri-Sen:96,Ghosh-Rao:98,Calogero:99}.
	Also the case of critical dimension, which for the standard
	bosonic case is equal to two, is covered, 
	and corresponding 
	inequalities involving logarithms found. 
	We point out 
	that some of the generalizations presented 
	have also been considered in the context of Hardy inequalities
	by Laptev et. al. \cite{Laptev}.
	For purposes of illustrating the method and techniques involved
	in a more accessible way, we choose to start from a simple setting
	and build up a more general framework rather than to start with
	the most general (but technically complicated) theorem 
	and then specialize from that.
	
	The main purpose of this work is twofold.
	First, the new bounds which are derived in 
	Theorems \ref{thm:parallel-no3d}-\ref{thm:volume-all} prove that 
	the classes of models considered in \cite{Murthy-Bhaduri-Sen:96,Ghosh-Rao:98} are, 
	or can be made, well-defined independently of external potentials for a 
	corresponding range of coupling parameters, in that the 
	corresponding quadratic forms are bounded from below. 
	Furthermore, the ground state representations given here will be of use 
	in the spectral analysis of the associated operators 
	(compare e.g. the analysis of natural self-adjoint extensions for the 
	Calogero-Sutherland operators in \cite{Lundholm-Solovej:extended}, 
	for which knowledge of the ground state representation was essential). 
	However, the study of optimal constants in the large-N limit is 
	unfortunately very complicated due to the geometrically complicated 
	nature of the potentials and will have to be addressed in future work,
	but we do derive the sharp constants for the constituent potentials with
	a fixed number of particles in the interaction (see Appendix B).

	Second, although the ground state approach employed here 
	is well known in the literature on Hardy inequalities, 
	it is hoped that the structure 
	of the paper and its systematic study of such inequalities 
	will help bring out the usefulness of the technique to an even wider 
	community of mathematicians and physicists.
	In particular, in combination with the new approach employed here 
	using geometric algebra (whose usefulness is also well known 
	within a certain community, although
	the intersection between these two communities seems unfortunately 
	to be rather small \cite{Sobczyk:12})
	the emphasis is on the fact that both technically and geometrically 
	complicated results can be obtained in an efficient way.
	The geometrically most general form, 
	the Corollary to Theorem \ref{thm:volume-all} 
	which shows that a class of generalized particle interactions involving 
	volumes can be 
	considered as small perturbations of the free kinetic energy operator,
	would probably not have been manageable without these tools and 
	the systematic exposition.
	
	The paper is organised as follows.
	In Section \ref{sec:preliminaries} we state the preliminary setup
	which allows 
	for a straightforward and systematic derivation of the results.
	The main results are given in Section \ref{sec:many-particle}
	as 
	ground state representations for the conventional 
	many-particle Hardy inequalities in all dimensions
	(Theorems \ref{thm:many-particle-separation}-\ref{thm:many-particle-Hardy-2d},
	with an extension in Theorem \ref{thm:many-particle-Hardy-combined}),
	for some alternative geometric cases where 
	the origin is singled out as a special point
	(Theorems \ref{thm:many-particle-1d-product}-\ref{thm:parallel-3d}),
	and finally for inequalities involving volumes of simplices
	of points (Theorems \ref{thm:volume-noncritical}-\ref{thm:volume-all}).
	Conclusions are given in Section \ref{sec:conclusions}.
	Some computations involving multilinear algebra,
	and a brief note 
	on the sharpness of the derived constants,
	have been placed in an appendix.

\begin{ack}
	I am thankful to Fabian Portmann for 
	valuable feedback and 
	for collaboration on closely related
	subjects. I would also like to thank 
	Bergfinnur Durhuus, Jens Hoppe, Ari Laptev, Enno Lenzmann, and Jan Philip Solovej 
	for discussions.
	Financial support from the Danish Council for Independent Research,
	Nordita, and the Isaac Newton Institute (EPSRC Grant EP/F005431/1)
	is gratefully acknowledged. 
	The majority of this work was carried out at the University of Copenhagen.
\end{ack}

\section{Preliminaries: single-particle Hardy inequalities}
	\label{sec:preliminaries}

	As a preparation,
	we start by recalling the ground state representation
	for the Laplacian
	in a form which is well suited for 
	our applications, and 
	use it to derive the standard single-particle Hardy 
	inequalities w.r.t. a point and 
	a higher-dimensional subspace in $\R^d$.

\subsection{The ground state representation}

	We have the following simple but general 
	\emph{ground state representation} (GSR)
	for the (Dirichlet) Laplacian on a domain in $\mathbb{R}^n$:

\begin{prop}[GSR] \label{prop:GSR}
	Let $\Omega$ be an open set in $\mathbb{R}^n$ and
	let $f: \Omega \to \mathbb{R}_+ := (0,\infty)$ be twice differentiable. 
	Then, for any $u \in C_0^\infty(\Omega)$ and $\alpha \in \mathbb{R}$,
	\begin{equation}
		\int_\Omega |\nabla u|^2 \,dx
		= \int_\Omega \left( 
				\alpha(1 - \alpha) \frac{|\nabla f|^2}{f^2} + \alpha\frac{-\Delta f}{f}
			\right) |u|^2 \,dx
		+ \int_\Omega |\nabla v|^2 f^{2\alpha} \,dx, 
		\label{GSR}
	\end{equation}
	where $v := f^{-\alpha}u$. 
\end{prop}
\begin{proof}
	We have for $u = f^{\alpha}v$ that
	$
		\nabla u = \alpha f^{\alpha-1} (\nabla f) v
		+ f^{\alpha} \nabla v,
	$
	hence
	$$
		|\nabla u|^2 = \alpha^2 f^{2(\alpha-1)} |\nabla f|^2 |v|^2
		+ \alpha f^{2\alpha-1} (\nabla f) \cdot \nabla |v|^2 
		+ f^{2\alpha} |\nabla v|^2.
	$$
	Integrating this expression over $\Omega$, we find that the middle term 
	on the r.h.s. produces after partial integration
	$$
		- \alpha \int_{\Omega} \nabla \cdot (f^{2\alpha-1} \nabla f) |v|^2 \,dx.
	$$
	Now, using that
	$$
		\nabla \cdot (f^{2\alpha-1} \nabla f) 
		= (2\alpha-1) f^{2\alpha-2} |\nabla f|^2 + f^{2\alpha-1} \Delta f,
	$$
	and collecting the terms we arrive at \eqref{GSR}.
\end{proof}

	It will in the sequel
	be very convenient to introduce some terminology related to the
	ground state representation.
	We refer 
	to $f$ as the (exact or approximate) \emph{ground state} 
	and to $\alpha$ as the \emph{GSR weight}, 
	while the potential term arising in \eqref{GSR}, i.e.
	\begin{equation} \label{GSR-potential}
		\alpha(1 - \alpha) \frac{|\nabla f|^2}{f^2} + \alpha\frac{-\Delta f}{f},
	\end{equation}
	will be called the \emph{GSR potential}.
	Note that the choice $\alpha = \frac{1}{2}$ maximizes the first term in the
	GSR potential,
	which would be the relevant term if
	$\Delta f = 0$ on $\Omega$,
	i.e. if $f$ is a generalized zero eigenfunction 
	for the Laplacian on $\Omega$.
	In this case the resulting GSR \eqref{GSR} 
	will usually be called a \emph{Hardy GSR},
	in anticipation of Hardy-type inequalities.
	A general idea of this approach, however, 
	is to try to find as good an inequality as possible by considering
	also \emph{approximate} ground states, 
	say of a particular form convenient for computations, 
	with the possibility to optimize over the GSR weight $\alpha$.
	We further emphasize 
	that an important advantage of having the ground state
	representation for a Hardy inequality
	is that the integral term involving $v$ in \eqref{GSR} 
	provides a guide for proving sharpness (cp. Appendix B), 
	and also opens up
	for further improvements of the inequality.

	We will for the following also find it useful to 
	note that a simple modification of Proposition \ref{prop:GSR}
	to involve a product ground state ansatz $g^\alpha h^\beta$
	(i.e. $u = g^\alpha h^\beta v$)
	produces the GSR potential 
	\begin{equation} \label{product-GSR}
		\alpha(1-\alpha) \frac{|\nabla g|^2}{g^2} + \alpha \frac{-\Delta g}{g}
		+ \beta(1-\beta) \frac{|\nabla h|^2}{h^2} + \beta \frac{-\Delta h}{h}
		- 2\alpha\beta \frac{\nabla g \cdot \nabla h}{gh}.
	\end{equation}

\subsection{The standard Hardy inequalities in $\R^d$}

	The obvious ground states $f$ for the Laplacian in $\R^d$
	are the \emph{fundamental solutions},
	\begin{eqnarray*}
		& f_{d \neq 2}(x) := |x|^{-(d-2)}, & \qquad \Delta_{\R^d} f_d = c_d \delta_0, \\
		& f_2(x) := \ln |x|, & \qquad \Delta_{\R^2} f_2 = c_2 \delta_0,
	\end{eqnarray*}
	where $\delta_0$ are Dirac delta distributions 
	supported at the origin 
	and $c_d$ some irrelevant constants.
	Hence, for $d \neq 2$ we can consider the domain 
	$\Omega := \R^d \setminus \{0\}$ on which $f := f_d > 0$ and $\Delta f = 0$.
	\eqref{GSR} is therefore optimal for $\alpha = \frac{1}{2}$, and
	yields the ground state 
	representation associated to the standard Hardy inequality 
	\eqref{standard-Hardy-ineq}
	in $\mathbb{R}^d$:
	\begin{equation} \label{standard-Hardy}
		\int_\Omega |\nabla u|^2 \,dx - \frac{(d-2)^2}{4} \int_\Omega \frac{|u|^2}{|x|^2} \,dx
		= \int_\Omega |\nabla v|^2 |x|^{-(d-2)} \,dx \ \ge \ 0.
	\end{equation}
	The inequality \eqref{standard-Hardy} 
	holds for all $u \in C_0^\infty(\Omega)$,
	and hence the l.h.s. is non-negative on the Sobolev space $H^1_0(\Omega)$
	($= H^1(\R^d)$ for $d \ge 2$) by closure.
	
	For $d=2$ we can take the domain
	$\Omega := \R^2 \setminus (\{0\} \cup \mathbb{S}^{1})$ 
	and ground state $f:=|f_2|$, so that $f>0$ and $\Delta f = 0$ on $\Omega$.
	\eqref{GSR} then produces the corresponding two-dimensional Hardy GSR
	\begin{equation} \label{standard-Hardy-2d}
		\int_\Omega |\nabla u|^2 \,dx - \frac{1}{4} \int_\Omega \frac{|u|^2}{|x|^2 (\ln |x|)^2} \,dx
		= \int_\Omega |\nabla v|^2 \big| \ln |x| \big| \,dx \ \ge \ 0.
	\end{equation}
	By closure, the l.h.s. is non-negative
	for $u \in H^1_0(\Omega) = H^1_0(\R^2 \setminus \mathbb{S}^1)$.

	In the above we used the standard fact that
	$C_0^\infty(\R^n \setminus \{0\})$ 
	is dense
	in $H^1(\R^n)$
	(with the Sobolev norm)
	for $n \ge 2$,
	while the closure of
	$C_0^\infty(\R \setminus \{0\})$ is 
	$H_0^1(\R \setminus \{0\}) \subsetneq H^1(\R)$.
	See e.g. Lemma 3 in \cite{Lundholm-Solovej:extended}
	for an explicit proof in the case of critical dimension $n=2$.
	Similar density arguments
	are valid for $C_0^\infty(\R^n \setminus K)$ for the codimension $k \ge 2$
	subsets $K$ we consider in the following, 
	typically being finite unions of closed 
	smooth or cone-like submanifolds of dimension $n-k$
	(see e.g. Section 9 in \cite{Mazya:85}).

\subsection{Hardy inequalities outside subspaces}

	Let us also briefly recall 
	the generalizations 
	of the standard Hardy inequalities 
	\eqref{standard-Hardy}-\eqref{standard-Hardy-2d} 
	w.r.t. the point $\{0\}$,
	to corresponding inequalities w.r.t. any linear subspace in $\R^d$.
	For later purposes,
	it is most convenient to state and derive these GSR in the
	language of \emph{geometric algebra} 
	(involving the Clifford algebra over $\mathbb{R}^d$; 
	see \cite{Lundholm-Svensson:09}, 
	or the brief introduction in Appendix A).
	
	Let $A := \ba_1 \wedge \ldots \wedge \ba_p \neq 0$ be a $p$-blade, 
	$0 \le p < d$, i.e. an exterior product of $p$ vectors 
	$\ba_j \in \R^d$,
	representing the oriented $p$-dimensional linear subspace 
	$\bar{A} \subseteq \R^d$ spanned by $\{\ba_j\}$, 
	with magnitude $|A|$.
	Note that the $p+1$-blade 
	$\bx \wedge A = 0$ if and only if $\bx \in \bar{A}$.
	Further,
	$\delta(\bx) := |\bx \wedge A||A|^{-1}$ is the minimal distance from 
	$\bx$ to $\bar{A}$,
	while 
	the Clifford product
	$(\bx \wedge A)A^{-1} = (1 - P_{\bar{A}})\bx$,
	where $P_{\bar{A}}$ is the orthogonal projection on $\bar{A}$.
	We then have the following simple Hardy GSRs, 
	which reduce to \eqref{standard-Hardy} resp. 
	\eqref{standard-Hardy-2d} for $p=0$,
	and which will shortly be generalized to many-particle versions:

\begin{thm}[$d-p \neq 2$] \label{thm:subspace-Hardy}
	Let $\Omega := \{ \bx \in \R^d : | \bx \wedge A | > 0 \}$.
	Taking the ground state 
	$f(\bx) := |\bx \wedge A|^{-(d-p-2)} \propto \delta(\bx)^{-(d-p-2)}$ 
	we obtain
	\begin{equation} \label{subspace-Hardy}
		\int_\Omega |\nabla u|^2 \,dx - \frac{(d-p-2)^2}{4} \int_\Omega \frac{|A|^2}{|\bx \wedge A|^2} |u|^2 \,dx
		= \int_\Omega |\nabla f^{-\frac{1}{2}}u|^2 f \,dx \ \ge \ 0,
	\end{equation}
	for $u \in C_0^\infty(\Omega)$.
	The corresponding Hardy inequality for the l.h.s. holds for
	$u \in H_0^1(\Omega)$ ($= H^1(\R^d)$ for $d-p \ge 2$).
\end{thm}
\begin{proof}
	By choosing a basis and coordinate system appropriately, 
	one easily computes that
	$\Delta \delta^{-(d-p-2)} = 0$ on $\Omega$
	($f$ is a fundamental solution to the Laplacian w.r.t. the subspace $\bar{A}$).
	Furthermore, $\nabla f/f = -(d-p-2)\nabla \delta/\delta$, 
	and it follows that the optimal weight is the standard Hardy 
	$\alpha = \frac{1}{2}$, with $|\nabla f|^2/f^2 = (d-p-2)^2/\delta^2$.

	Alternatively, by employing geometric algebra we can avoid
	introducing coordinates and directly obtain
	$\nabla f = -(d-p-2)(\bx \wedge A)^{-1}A^\dagger f$
	and $\Delta f = 0$
	(see Appendix A).
\end{proof}

\begin{thm}[$d-p = 2$] \label{thm:subspace-Hardy-2d}
	Fix a length scale $R > 0$ and consider
	$\Omega := \{ \bx \in \R^d : 0 < | \bx \wedge A |/R \neq 1 \}$.
	Taking $f(\bx) := \big| \ln \frac{1}{R}|\bx \wedge A| \big|$ we obtain
	\begin{equation} \label{subspace-Hardy-2d}
		\int_\Omega |\nabla u|^2 \,dx - \frac{1}{4} \int_\Omega 
			\frac{|A|^2}{|\bx \wedge A|^2 (\ln \frac{1}{R}|\bx \wedge A|)^2} |u|^2 \,dx
		= \int_\Omega |\nabla f^{-\frac{1}{2}}u|^2 f \,dx \ \ge \ 0,
	\end{equation}
	for $u \in C_0^\infty(\Omega)$.
\end{thm}
\begin{proof}
	Here we have
	$\nabla f = \pm(\bx \wedge A)^{-1}A$ 
	(with a sign depending on $p$ and which part of $\Omega$ we consider)
	and
	$\Delta f = \pm (d-p-2)|A|^2|\bx \wedge A|^{-2} = 0$
	(see Appendix A).
\end{proof}

	Due to the invariance of the Laplacian under translations,
	corresponding GSR of course hold also for affine subspaces
	$\ba + \bar{A}$, $\ba \in \R^d$,
	simply by translation $\bx \mapsto \bx + \ba$
	and considering
	$\Omega := \{ |(\bx - \ba) \wedge A| > 0 \}$
	(analogously for $d-p=2$).
	Furthermore, we note that the constants in \eqref{subspace-Hardy}
	are sharp, just as for $p=0$ (cp. Appendix B).

\section{Many-particle Hardy inequalities}
	\label{sec:many-particle}

	We now turn to a 
	systematic application of the GSR 
	with approximate ground states
	to the setting of many-particle Hardy inequalities.

\subsection{Conventional many-particle inequalities}

	Consider a tuple $(\bx_1,\ldots,\bx_N)$ of $N$ points, or \emph{particles}, 
	in $\mathbb{R}^d$.
	We define the distance $r_{ij} := |\bx_i - \bx_j|$ between two particles,
	and the \emph{circumradius} $R_{ijk}$ 
	associated to three non-coincident particles,
	$$
		\frac{1}{2R_{ijk}^2} := \sum_{\textrm{cyclic in $i,j,k$}} (\bx_i - \bx_j)^{-1} \cdot (\bx_i - \bx_k)^{-1},
	$$
	i.e. the radius of the circle that the particles 
	$\bx_i$, $\bx_j$, $\bx_k$ inscribe
	($R_{ijk} := \infty$ for collinear particles,
	for which the r.h.s. is zero; 
	cp. Lemma 3.2 in \cite{H-O2-Laptev-Tidblom:08}).
	Let us first consider the total separation measured by the
	distance-squared between all pairs of particles:

\begin{thm}[Total separation of $N\ge 2$ particles] 
	\label{thm:many-particle-separation}
	Let
	$$
		\Omega := \R^{dN} \setminus \{ \bx_1 = \bx_2 = \ldots = \bx_N \}.
	$$
	Taking the ground state $\rho(x)^2 := \sum_{i<j} |\bx_i - \bx_j|^{2}$ we obtain
	\begin{equation} \label{many-particle-separation}
		\int_\Omega |\nabla u|^2 \,dx
		- N \left( \frac{N-1}{2}d - 1 \right)^2 \int_\Omega 
			\frac{|u|^2}{\rho^2} \,dx
		= \int_\Omega |\nabla \rho^{-2\alpha} u|^2 \rho^{4\alpha} \,dx \ge 0,
	\end{equation}
	for all $u \in C_0^\infty(\Omega)$,
	with the optimal weight $\alpha := -\frac{(N-1)d - 2}{4}$.
\end{thm}
\noindent
	This gives a generalization to $N \neq 3$ 
	of (3.5) in \cite{H-O2-Laptev-Tidblom:08}.
	Note that $\codim \Omega^c = dN-d = d(N-1)$
	and hence that the corresponding 
	Hardy inequality
	on $H^1_0(\Omega)$ holds on $H^1(\R^{dN})$
	unless $d=1$ and $N=2$.
	The constant is sharp, as shown explicitly in the appendix
	(Proposition \ref{prop:sharpness-mp-separation}).
\begin{proof}
	One computes
	$$
		\nabla_k \rho^2 = 2 \sum_{j \neq k} (\bx_k - \bx_j),
	$$
	hence
	$$
		\Delta \rho^2 = 2 \sum_k \sum_{j \neq k} \nabla_k \cdot (\bx_k - \bx_j) = 2N(N-1)d,
	$$
	and
	$$
		|\nabla \rho^2|^2 = 8 \sum_{i<j} r_{ij}^2
			+ 8 \sum_{k} \sum_{i<j} (\bx_k - \bx_i) \cdot (\bx_k - \bx_j)
			= 4N\rho^2,
	$$
	where in the last step we used the identity
	\begin{equation} \label{particle-sum-identity}
		\sum_{k} \sum_{i<j} (\bx_k - \bx_i) \cdot (\bx_k - \bx_j) 
			= \frac{N-2}{2} \sum_{i<j} |\bx_i - \bx_j|^2.
	\end{equation}
	It follows that we have a GSR potential
	$$
		\alpha(1-\alpha) \frac{|\nabla \rho^2|^2}{\rho^4} - \alpha \frac{\Delta \rho^2}{\rho^2}
		= 4N \alpha\left( \frac{2 - (N-1)d}{2} - \alpha \right) \frac{1}{\rho^2},
	$$
	which by optimization proves the theorem.
\end{proof}

	Next, we have as a special case of the following, 
	the so-called
	`standard' many-particle Hardy inequality:

\begin{thm}[Separation and circumradii of pairs and triples of particles in $d \ge 3$] 
	\label{thm:many-particle-Hardy-3d}
	Let 
	\begin{equation} \label{omega-minus-diagonals}
		\Omega := \{ (\bx_1,\ldots,\bx_N) \in \R^{dN} : \bx_i \neq \bx_j \ \forall i \neq j \}.
	\end{equation}
	Taking the ground state $f(x) := \prod_{j<k} |\bx_j - \bx_k|^{-(d-2)}$ we obtain
	\begin{equation}
	\begin{split}
		\int_\Omega |\nabla u|^2 \,dx
		- (d-2)^2 \int_\Omega \left( 
			2\alpha(1-\alpha) \sum_{i<j} \frac{1}{r_{ij}^2} 
			- \alpha^2 \sum_{i<j<k} \frac{1}{R_{ijk}^2} 
			\right) |u|^2 \,dx
			\\
		= \int_\Omega |\nabla f^{-\alpha} u|^2 f^{2\alpha} \,dx  \ \ge \ 0,
	\end{split}
	\label{many-particle-Hardy-3d-alpha}
	\end{equation}
	for all $u \in C_0^\infty(\Omega)$ and $\alpha \in \R$.
	In particular, defining 
	\begin{equation} \label{K-bound}
		K_{d,N} := \sup_{x \in \Omega} 
			\frac{ \sum_{i<j<k} 1/R_{ijk}^2 }{ \sum_{i<j} 1/r_{ij}^2 }
			\quad \le N-2 < \infty,
	\end{equation}
	(the upper bound following from geometric relations; 
	cp. Lemma 3.3 in \cite{H-O2-Laptev-Tidblom:08}),
	and taking the then optimal weight $\alpha := \frac{1}{2 + K_{d,N}}$ we have
	\begin{equation} \label{many-particle-Hardy-3d}
		\int_\Omega |\nabla u|^2 \,dx
		- \frac{(d-2)^2}{2 + K_{d,N}} \sum_{i<j} \int_\Omega 
			\frac{|u|^2}{r_{ij}^2} \,dx
		\ge \int_\Omega |\nabla f^{-\alpha} u|^2 f^{2\alpha} \,dx  \ \ge \ 0.
	\end{equation}
	On the other hand, assuming $\alpha(1-\alpha) \ge 0$ in 
	\eqref{many-particle-Hardy-3d-alpha}
	and now using the bound \eqref{K-bound} on the term involving $r_{ij}$
	we obtain
	\begin{equation} \label{many-particle-Hardy-3d-R}
		\int_\Omega |\nabla u|^2 \,dx
		- \frac{(d-2)^2}{K_{d,N}(2 + K_{d,N})} \sum_{i<j<k} \int_\Omega 
			\frac{|u|^2}{R_{ijk}^2} \,dx
		\ge \int_\Omega |\nabla f^{-\alpha} u|^2 f^{2\alpha} \,dx  \ \ge \ 0,
	\end{equation}
	again with the optimal weight $\alpha := \frac{1}{2 + K_{d,N}}$.
\end{thm}

	As pointed out in \cite{H-O2-Laptev-Tidblom:08}, 
	using the geometric relations between separation and circumradii,
	the $N=3$ case of Theorem \ref{thm:many-particle-separation} also implies
	\begin{equation} \label{many-particle-Hardy-3d-3}
		\int_\Omega |\nabla u|^2 \,dx
		\ge \frac{(d-1)^2}{3} \binom{N}{3}^{-1} \frac{N}{3} \sum_{i<j<k} \int_\Omega 
			\frac{|u|^2}{R_{ijk}^2} \,dx.
	\end{equation}
	Combining this with \eqref{many-particle-Hardy-3d-alpha}, 
	one is led to maximize $\frac{\alpha(1-\alpha)}{1 + c\alpha^2}$
	with $$c := \frac{3}{2}\frac{(d-2)^2}{(d-1)^2}(N-1)(N-2).$$
	This results in
	\begin{equation} \label{many-particle-Hardy-3d-c}
		\int_\Omega |\nabla u|^2 \,dx
		- \alpha (d-2)^2
		\sum_{i<j} \int_\Omega \frac{|u|^2}{r_{ij}^2} \,dx
		\ge \frac{1}{1 + c\alpha^2} \int_\Omega |\nabla f^{-\alpha} u|^2 f^{2\alpha} \,dx  \ \ge \ 0,
	\end{equation}
	with the optimal $\alpha := (1 + \sqrt{1+c})^{-1}$.
	The Hardy inequalities corresponding to 
	\eqref{many-particle-Hardy-3d} and \eqref{many-particle-Hardy-3d-c}
	were given in the form of
	Theorem 2.1, (4.9), (4.11) in \cite{H-O2-Laptev-Tidblom:08}.

	We also note that
	all corresponding Hardy inequalities from
	\eqref{many-particle-Hardy-3d}--\eqref{many-particle-Hardy-3d-c}
	hold on the full space of functions $H^1(\R^{dN})$
	since $\codim \Omega^c = dN - d - (N-2)d = d \ge 3$.
	For $N=2$ we simply have $\alpha = \frac{1}{2}$ and
	\eqref{many-particle-Hardy-3d} and \eqref{many-particle-Hardy-3d-c}
	reduce to \eqref{many-particle-separation}
	with the sharp constant $(d-2)^2/2$.
	It was also noted in \cite{H-O2-Laptev-Tidblom:08}
	that the large-$N$ behavior of the constant in
	\eqref{many-particle-Hardy-3d}/\eqref{many-particle-Hardy-3d-c}
	with $K_{d,N} \sim N$ 
	cannot be improved.

\begin{proof}
	One computes
	$$
		\nabla_k f = -(d-2)f \sum_{j \neq k} (\bx_k - \bx_j)^{-1},
	$$
	hence
	$$
		|\nabla f|^2 = (d-2)^2 f^2 \left( 2\sum_{i<j} \frac{1}{r_{ij}^2} + \sum_{i<j<k} \frac{1}{R_{ijk}^2} \right),
	$$
	and, due to $\Delta_{\bx_k} |\bx_k - \bx_j|^{-(d-2)} = 0$ on $\Omega$,
	$$
		\Delta f = (d-2)^2 f \sum_{i<j<k} \frac{1}{R_{ijk}^2}.
	$$
	This gives the GSR \eqref{many-particle-Hardy-3d-alpha} in the theorem.
	Bounding the (in total positive) term involving $R_{ijk}$ in that equation
	by the term involving $r_{ij}$ by means of $K_{d,N}$ in \eqref{K-bound}, 
	we obtain the total constant in \eqref{many-particle-Hardy-3d}
	$$
		(d-2)^2\alpha(2(1 - \alpha) - \alpha K_{d,N}) 
		= (d-2)^2(2 + K_{d,N}) \, \alpha\left( \frac{2}{2+K_{d,N}} - \alpha \right),
	$$
	which is optimal for $\alpha = (2+K_{d,N})^{-1}$.
	\eqref{many-particle-Hardy-3d-R} follows similarly.
\end{proof}

	The one-dimensional case, on the other hand, 
	is much simpler due to collinearity of the particles:

\begin{thm}[Separation of pairs of particles in $d=1$] 
	\label{thm:many-particle-Hardy-1d}
	Let
	$$
		\Omega := \{ (x_1,\ldots,x_N) \in \R^N : x_i \neq x_j \ \forall i \neq j \}.
	$$
	Taking the ground state $f(x) := \prod_{j<k} |x_j - x_k|$ we obtain
	\begin{equation} \label{many-particle-Hardy-1d}
		\int_\Omega |\nabla u|^2 \,dx
		- \frac{1}{2} \int_\Omega \left( \sum_{i<j} \frac{1}{r_{ij}^2} \right) |u|^2 \,dx
		= \int_\Omega |\nabla f^{-\frac{1}{2}} u|^2 \prod_{i<j} r_{ij} \,dx  \ \ge \ 0,
	\end{equation}
	for all $u \in C_0^\infty(\Omega)$.
\end{thm}

\begin{proof}
	In this case $R_{ijk} = \infty$ and $\Delta f = 0$ on $\Omega$.
	Hence $\alpha = \frac{1}{2}$ optimizes the GSR.
\end{proof}

	This is a GSR version of \eqref{standard-1D-many-particle} 
	and Theorem 2.5 in \cite{H-O2-Laptev-Tidblom:08}.
	The corresponding inequality holds for all
	$u \in H^1_0(\Omega)$ 
	(note that $\codim \Omega^c = 1$ in this case)
	and is sharp.
	This lower bound \eqref{many-particle-Hardy-1d} 
	(as well as a corresponding identity for $\alpha > 1/2$)
	plays an important role for operators appearing in 
	the Calogero-Sutherland models 
	\cite{Calogero:69,Calogero:71,Sutherland:71}, 
	and for a model of identical particles in one dimension 
	with generalized statistics 
	\cite{Polychronakos:89,Polychronakos:92,Lundholm-Solovej:extended}.

	For the two-dimensional case we fix a length scale $R>0$ and define
	$$
	\tilde{r}_{ij} := |\bx_i - \bx_j| \big| \ln \frac{1}{R}|\bx_i - \bx_j| \big|,
	$$
	and $\tilde{R}_{ijk}$ by
	\begin{equation} \label{def-R-tilde}
		\frac{1}{2\tilde{R}_{ijk}^2} := \sum_{\textrm{cycl}} 
			\frac{(\bx_i - \bx_j)^{-1}}{\big| \ln \frac{1}{R}|\bx_i - \bx_j| \big|} \cdot 
			\frac{(\bx_i - \bx_k)^{-1}}{\big| \ln \frac{1}{R}|\bx_i - \bx_k| \big|}.
	\end{equation}

\begin{thm}[Separation of pairs of particles in $d=2$] 
	\label{thm:many-particle-Hardy-2d}
	Let
	$$
		\Omega := \{ (\bx_1,\ldots,\bx_N) \in \R^{2N} : \bx_i \neq \bx_j \ \forall i \neq j \} \cap (B_{R/2}(0))^N.
	$$
	Taking the ground state 
	$f(x) := \prod_{i<j} \big| \ln \frac{1}{R} |\bx_i - \bx_j| \big|$ we obtain
	\begin{eqnarray*}
		\int_\Omega |\nabla u|^2 \,dx
		- \int_\Omega \left( 
			2\alpha(1-\alpha) \sum_{i<j} \frac{1}{\tilde{r}_{ij}^2} 
			- \alpha^2 \sum_{i<j<k} \frac{1}{\tilde{R}_{ijk}^2} 
			\right) |u|^2 \,dx
			\\
		= \int_\Omega |\nabla f^{-\alpha} u|^2 f^{2\alpha} \,dx  \ \ge \ 0,
	\end{eqnarray*}
	for all $u \in C_0^\infty(\Omega)$.
	Hence, if 
	$K_{2,N} := \sup_{x \in \Omega} \sum \tilde{R}^{-2}_{ijk} / \sum \tilde{r}^{-2}_{ij}$,
	then
	\begin{eqnarray} \label{many-particle-Hardy-2d}
		\int_\Omega |\nabla u|^2 \,dx
		- \frac{1}{2 + K_{2,N}} \sum_{i<j} \int_\Omega 
			\frac{|u|^2}{\tilde{r}_{ij}^2} 
			\,dx
		\ge \int_\Omega |\nabla f^{-\alpha} u|^2 f^{2\alpha} \,dx  \ \ge \ 0,
	\end{eqnarray}
	and
	\begin{equation} \label{many-particle-Hardy-2d-R}
		\int_\Omega |\nabla u|^2 \,dx
		- \frac{1}{K_{2,N}(2 + K_{2,N})} \sum_{i<j<k} \int_\Omega 
			\frac{|u|^2}{\tilde{R}_{ijk}^2} \,dx
		\ge \int_\Omega |\nabla f^{-\alpha} u|^2 f^{2\alpha} \,dx  \ \ge \ 0,
	\end{equation}
	with $\alpha := \frac{1}{2 + K_{2,N}}$.
\end{thm}

\begin{proof}
	This theorem follows just as in the proof of 
	Theorem \ref{thm:many-particle-Hardy-3d},
	using that
	$\nabla_k f = f \sum_{j \neq k} \big| \ln \frac{1}{R} |\bx_k - \bx_j| \big|^{-1} (\bx_k - \bx_j)^{-1}$
	and
	$\Delta f = f \sum_{i<j<k} \tilde{R}_{ijk}^{-2}$.
\end{proof}

	Here we have a rough bound
	$K_{2,N} \le 2(N-2)$, 
	which simply follows by applying the Cauchy-Schwarz inequality 
	in $\R^2$ to each term in \eqref{def-R-tilde}.

	We have
	from \eqref{many-particle-Hardy-2d} a 
	non-trivial two-dimensional many-particle Hardy inequality 
	on $H_0^1(\Omega) = H_0^1(B_{R/2}(0)^N)$.
	This correponds to the physical situation where
	we consider the particles to be confined to a finite area.
	Also note that, despite the logarithms,
	this inequality gives a rough lower bound
	$$
		\inf_{\substack{u \in H_0^1(\Omega) \,:\\ \|u\|_{L^2} = 1}} 
		\int_{B_{R/2}(0)^N} |\nabla u|^2 \,dx 
		\ \ge \ \frac{\textrm{const} \cdot \binom{N}{2}}{(2+K_{2,N})R^2} 
	$$
	for the ground state energy of a confined gas of 
	two-dimensional non-inter\-act\-ing 
	bosonic particles, which hence is of the same form as for $d \ge 3$.

	As an illustration of the
	freedom for improvement left in the remainder terms in the above GSRs,
	we show that
	it is e.g. possible to combine and improve 
	Theorem \ref{thm:many-particle-separation}
	with 
	Theorem \ref{thm:many-particle-Hardy-3d}:

\begin{thm}[Combined pairwise separation and total separation]
	\label{thm:many-particle-Hardy-combined}
	Taking $\Omega$ as in \eqref{omega-minus-diagonals} and the ground state
	$f(x) := \rho^{2\frac{\beta}{\alpha}} \prod_{i<j}r_{ij}^{-(d-2)}$ 
	with $\alpha = (2+K_{d,N})^{-1}$ 
	and $\beta = \frac{1}{4}((\alpha N (d-2) - d)(N-1) + 2)$,
	we obtain
	\begin{multline} \label{many-particle-Hardy-combined}
		\int_{\R^{dN}} |\nabla u|^2 \,dx
		\ge \int_{\R^{dN}} \left( 
			\frac{(d-2)^2}{2 + K_{d,N}} \sum_{i<j} \frac{1}{r_{ij}^2} \,dx \right. \\
			\left. + \ \frac{N}{4}\left( 
					\left( \frac{N(d-2)}{2+K_{d,N}} - d \right)(N-1) + 2 
				\right)^2 \frac{1}{\rho^2}
			\right) |u|^2 \,dx
	\end{multline}
	for all $u \in H^1(\R^{dN})$, $d \ge 3$.
\end{thm}
\begin{proof}
	Taking $g(x) := \prod_{i<j} r_{ij}^{-(d-2)}$ and $h(x) := \rho^2$
	in the product GSR potential \eqref{product-GSR},
	we find, together with the above computations,
	$$
		\nabla g \cdot \nabla h = -(d-2)N^2(N-1)g,
	$$
	which follows from an identity related to \eqref{particle-sum-identity},
	\begin{equation} \label{particle-inverse-identity}
		\sum_{k} \sum_{i \neq k} \sum_{j \neq k} 
			(\bx_k - \bx_i) \cdot (\bx_k - \bx_j)^{-1}
			= \frac{1}{2} N^2(N-1).
	\end{equation}
	This identity is most easily proved 
	(cp., e.g., Eqn. (3.9b) in \cite{Hoppe:92})
	by introducing the center-of-mass $\bX := \frac{1}{N}\sum_j \bx_j$
	and writing $\sum_{i\neq k} (\bx_k - \bx_i) = N(\bx_k - \bX)$,
	so that the l.h.s. of \eqref{particle-inverse-identity} becomes
	\begin{multline*}
		N \sum_{k \neq j} (\bx_k - \bX) \cdot (\bx_k - \bx_j)^{-1}
		= \\ \frac{N}{2} \left(
				\sum_{k \neq j} (\bx_k - \bX) \cdot (\bx_k - \bx_j)^{-1} 
				- \sum_{j \neq k} (\bx_j - \bX) \cdot (\bx_k - \bx_j)^{-1}
			\right)
		= \frac{N}{2} \sum_{j \neq k} 1.
	\end{multline*}
	Now, collecting the constants in front of the terms 
	in \eqref{product-GSR} involving $1/\rho^2$
	gives
	$$
		4N\beta\Big( \big( \alpha(d-2)N(N-1) - d(N-1) + 2 \big)/2 - \beta \Big),
	$$
	which with the earlier bound \eqref{K-bound} on the circumradius terms 
	and the corresponding choice of $\alpha$,
	and again by optimizing in the weight $\beta$,
	produces the statement of the theorem.
\end{proof}

\subsection{Some other inequalities of many-particle type}

	Associated to
	the original one-dimensional Hardy inequality away from
	the origin is also the following `many-particle' version,
	which can be viewed as a certain 
	limiting case of
	Theorem \ref{thm:many-particle-Hardy-combined} for $d=1$:

\begin{thm} 
	\label{thm:many-particle-1d-product}
	Let $\Omega := \{ (x_1,\ldots,x_N) \in \R^N : x_1 \ldots x_N \neq 0 \}$.
	Taking the ground state 
	$f(x) := |x|^{2(1-N)} \prod_{k=1}^N |x_k|$
	we obtain
	\begin{eqnarray*}
		\int_\Omega |\nabla u|^2 \,dx
		- \int_\Omega \left( 
			\frac{1}{4} \sum_{k=1}^N \frac{1}{x_k^2} 
			+ (N-1)^2 \frac{1}{|x|^2}
			\right) |u|^2 \,dx \\
		= \int_\Omega |\nabla f^{-\frac{1}{2}} u|^2 \prod_{k=1}^N |x_k| \, |x|^{2(1-N)} \,dx  \ \ge \ 0,
	\end{eqnarray*}
	for all $u \in C_0^\infty(\Omega)$.
	Hence, the l.h.s. is non-negative on $H_0^1(\Omega)$.
\end{thm}
\noindent
	The corresponding Hardy inequality was proved and
	applied in \cite{Tidblom:thesis}.
\begin{proof}
	With $g := \prod_k |x_k|$ and $h := |x|^2$ in \eqref{product-GSR}, 
	and using that
	$\nabla_k g = g/x_k$, $\Delta g = 0$,
	$\nabla_k h = 2x_k$, and $\Delta h = 2N$,
	we find the potential
	$$
		\alpha(1-\alpha) \sum_k \frac{1}{x_k^2}
		+ 4\beta \left( \frac{2 - (1+2\alpha)N}{2} - \beta \right) \frac{1}{|x|^2},
	$$
	which with the condition that the first term be optimal
	yields $\alpha := \frac{1}{2}$ and $\beta := \frac{1-N}{2}$.
	We also note that $\codim \Omega^c = 1$.
\end{proof}

	The following application of our systematic approach 
	has some relations with 
	\emph{bosonic} and \emph{supersymmetric matrix models} 
	(see e.g. \cite{Lundholm:thesis}),
	and operators of the corresponding form also appear in
	a class of solvable models 
	\cite{Murthy-Bhaduri-Sen:96,Ghosh-Rao:98}.
	Consider an $N$-tuple $(\bx_1,\ldots,\bx_N)$ of vectors in $\mathbb{R}^d$ 
	and define the bivectors (2-blades)
	$B_{ij} := \bx_i \wedge \bx_j$, $i<j$, with magnitudes
	$$
		|B_{ij}| = \sqrt{|\bx_i|^2|\bx_j|^2 - (\bx_i \cdot \bx_j)^2} = |\bx_i||\bx_j|\sin \theta_{ij},
	$$
	and inverses $B_{ij}^{-1} = (\bx_i \wedge \bx_j)^{-1} = -(\bx_i \wedge \bx_j) / |B_{ij}|^{2}$
	(again, cp. \cite{Lundholm-Svensson:09} or Appendix A).
	Note that $B_{ij} = 0$ if and only if $\bx_i$ and $\bx_j$ are parallel.
	We define
	the corresponding geometric quantities
	$$
		\Sigma_1(x) := \sum_{j \neq k} \big| \bx_j \liprod (\bx_j \wedge \bx_k)^{-1} \big|^2
		= \sum_{j<k} \frac{|\bx_j|^2 + |\bx_k|^2}{|\bx_j \wedge \bx_k|^2},
	$$
	and
	$$
		\Sigma_2(x) := \sum_{i \neq j \neq k \neq i} 
			\left( \bx_i \liprod (\bx_i \wedge \bx_k)^{-1} \right) \cdot
			\left( \bx_j \liprod (\bx_j \wedge \bx_k)^{-1} \right),
	$$
	where we note that
	$$
		- \bx_i \liprod (\bx_i \wedge \bx_k)^{-1}
		= \bx_i (\bx_i \wedge \bx_k) |\bx_i \wedge \bx_k|^{-2}
	$$
	is the vector in the (oriented) plane $\bx_i \wedge \bx_k$
	obtained by rotating $\bx_i$ by 90$^\circ$ 
	towards $\bx_k$
	and rescaling by the inverse area $|\bx_i \wedge \bx_k|^{-1}$.
	We then have the following result, 
	which one could think of as a 
	higher-dimensional combination 
	of Theorem \ref{thm:subspace-Hardy}
	with Theorem \ref{thm:many-particle-Hardy-3d},
	involving 1-dim\-en\-sional subspaces ($A=\bx_j$) instead of 
	0-dim\-en\-sional ($A = 1$):

\begin{thm}[Parallelity of pairs of vectors in $d > 3$ or $d = 2$] 
	\label{thm:parallel-no3d}
	\ \\
	Let $\Omega := \{ (\bx_1,\ldots,\bx_N) \in \R^{dN} : \bx_i \wedge \bx_j \neq 0 \ \forall i \neq j \}$.
	Taking the ground state $f(x) := \prod_{j<k} |B_{jk}|^{-(d-3)}$ one obtains
	\begin{eqnarray*}
		\int_\Omega |\nabla u|^2 \,dx
		- (d-3)^2 \int_\Omega \left( 
			\alpha(1-\alpha) \Sigma_1 
			- \alpha^2 \Sigma_2 
			\right) |u|^2 \,dx
			\\
		= \int_\Omega |\nabla f^{-\alpha} u|^2 f^{2\alpha} \,dx  \ \ge \ 0,
	\end{eqnarray*}
	for all $u \in C_0^\infty(\Omega)$.
	Hence, if 
	$C_{d,N} := \sup_{x \in \Omega} \Sigma_2(x) / \Sigma_1(x)$,
	then
	\begin{eqnarray*}
		\int_\Omega |\nabla u|^2 \,dx
		- \frac{(d-3)^2}{4(1 + C_{d,N})} \int_\Omega 
			\Sigma_1
			|u|^2 \,dx
		\ge \int_\Omega |\nabla f^{-\alpha} u|^2 f^{2\alpha} \,dx  \ \ge \ 0,
	\end{eqnarray*}
	and,
	\begin{eqnarray*}
		\int_\Omega |\nabla u|^2 \,dx
		- \frac{(d-3)^2}{4C_{d,N}(1 + C_{d,N})} \int_\Omega \Sigma_2 |u|^2 \,dx
		\ge \int_\Omega |\nabla f^{-\alpha} u|^2 f^{2\alpha} \,dx  \ \ge \ 0,
	\end{eqnarray*}
	with $\alpha := \frac{1}{2(1 + C_{d,N})}$ (in both inequalities).
\end{thm}

	Note that here $C_{d,N} \le N-2$ 
	by the Cauchy-Schwarz inequality in $\R^d$.
	Furthermore,
	$\codim \Omega^c = dN - (N-2)d - d - 1 = d-1$,
	so the corresponding inequalities hold on $H^1(\R^{dN})$ for $d>2$
	and $H_0^1(\Omega)$ for $d=2$.
	For $N=2$ we have $\Sigma_2 = 0$ and hence the optimal and sharp 
	2-particle Hardy inequality
	\begin{equation} \label{2-particle-area-Hardy}
		\int_{\R^{2d}} |\nabla u|^2 \,dx 
		\ge \frac{(d-3)^2}{4} \int_{\R^{2d}} 
			\frac{|\bx_1|^2 + |\bx_2|^2}{|\bx_1 \wedge \bx_2|^2} 
			|u|^2 \,dx,
	\end{equation}
	for $u \in H^1(\R^{2d})$, $d>3$, or $u \in H^1_0(\Omega)$, $d=2$
	(cp. Theorem \ref{thm:subspace-Hardy}).

\begin{proof}
	One computes (see Appendix A)
	$$
		\nabla_k f = -\frac{(d-3)}{2} f \sum_{j \neq k} |B_{jk}|^{-2} \nabla_k |B_{jk}|^2
		= -(d-3) f \sum_{j \neq k} \bx_j \liprod (\bx_j \wedge \bx_k)^{-1},
	$$
	implying
	$$
		|\nabla f|^2 = (d-3)^2 f^2 (\Sigma_1 + \Sigma_2),
	$$
	as well as $\Delta f = (d-3)^2 f \Sigma_2$, 
	due to $\Delta_k |B_{kj}|^{-(d-3)} = 0$ on $\Omega$ $\forall j$.
\end{proof}

	For the critical case $d=3$ we once again need 
	a length scale.
	One natural way to obtain this in the context of matrix 
	models\footnote{The case with $d=3$ and $N$ particles
	corresponds in the matrix models to dimensionally reduced
	$\SU(2)$ Yang-Mills theory from $N+1$ spacetime dimensions.}
	is to introduce the \emph{matrix model potential}
	$W := \sum_{j<k} |B_{jk}|^2$
	and consider e.g.
	$\Omega_0 := \Omega \cap \{W<R\}$
	for which $H^1_0(\Omega_0) = H^1_0(\{W<R\})$.
	Defining
	$$
		\tilde{\Sigma}_1(x) 
		:= \sum_{j<k} \frac{|\bx_j|^2 + |\bx_k|^2}{|B_{jk}|^2 \big| \ln \frac{1}{R}|B_{jk}| \big|^2},
	$$
	and
	$$
		\tilde{\Sigma}_2(x) := \sum_{i \neq j \neq k \neq i} 
			\frac{ \bx_i \liprod B_{ik}^{-1} }{ \big| \ln \frac{1}{R}|B_{ik}| \big| } \cdot
			\frac{ \bx_j \liprod B_{jk}^{-1} }{ \big| \ln \frac{1}{R}|B_{jk}| \big| },
	$$
	we then have the following 
	analog of Theorem \ref{thm:parallel-no3d}.

\begin{thm}[Parallelity of pairs of vectors in $d = 3$] 
	\label{thm:parallel-3d}
	Taking the ground state 
	$f(x) := \prod_{j<k} |\ln \frac{1}{R}|B_{jk}||$ one obtains
	\begin{eqnarray*}
		\int_{\Omega_0} |\nabla u|^2 \,dx
		- \int_{\Omega_0} \left( 
			\alpha(1-\alpha) \tilde{\Sigma}_1 
			- \alpha^2 \tilde{\Sigma}_2 
			\right) |u|^2 \,dx
			\\
		= \int_{\Omega_0} |\nabla f^{-\alpha} u|^2 f^{2\alpha} \,dx  \ \ge \ 0,
	\end{eqnarray*}
	for all $u \in C_0^\infty(\Omega_0)$.
	Hence, if 
	$C_{3,N} := \sup_{x \in \Omega_0} \tilde{\Sigma}_2(x) / \tilde{\Sigma}_1(x)$
	($\le N-2$),
	then
	\begin{eqnarray*}
		\int_{\Omega_0} |\nabla u|^2 \,dx
		- \frac{1}{4(1 + C_{3,N})} \int_{\Omega_0} 
			\tilde{\Sigma}_1 |u|^2 \,dx
		\ge \int_{\Omega_0} |\nabla f^{-\alpha} u|^2 f^{2\alpha} \,dx  \ \ge \ 0,
	\end{eqnarray*}
	and,
	\begin{eqnarray*}
		\int_{\Omega_0} |\nabla u|^2 \,dx
		- \frac{1}{4C_{d,N}(1 + C_{3,N})} \int_{\Omega_0} 
			\tilde{\Sigma}_2 |u|^2 \,dx
		\ge \int_{\Omega_0} |\nabla f^{-\alpha} u|^2 f^{2\alpha} \,dx  \ \ge \ 0,
	\end{eqnarray*}
	with $\alpha := \frac{1}{2(1 + C_{3,N})}$ (in both inequalities).
	In particular, with $N=2$,
	\begin{equation} \label{2-particle-area-Hardy-3d}
		\int_{\R^{6}} |\nabla u|^2 \,dx
		\ge \frac{1}{4} \int_{\R^{6}} 
			\frac{|\bx_1|^2 + |\bx_2|^2}{|\bx_1 \wedge \bx_2|^2 \left( \ln \frac{1}{R}|\bx_1 \wedge \bx_2| \right)^2} 
			|u|^2 \,dx,
	\end{equation}
	for all $u \in H^1_0(\{W<R\})$.
\end{thm}
\noindent
	The proof is completely analogous to the proof of 
	Theorem \ref{thm:parallel-no3d}, with similar modifications as in 
	Theorem \ref{thm:many-particle-Hardy-2d}.

\subsection{Inequalities involving volumes of simplices of points}

	Consider again a tuple of $N$ vectors in $\mathbb{R}^d$, 
	but now think of them as points.
	These points span a (possibly degenerate) $N-1$-simplex with
	volume given by
	\begin{eqnarray*}
		V(\bx_1,\ldots,\bx_N) 
		&:=& \frac{1}{(N-1)!} \big| (\bx_{1} - \bx_{N}) \wedge \ldots \wedge (\bx_{N-1} - \bx_{N}) \big| \\
		&=& \frac{1}{(N-1)!} \big| \det [(\bx_j-\bx_N) \cdot (\bx_k-\bx_N)]_{1 \le j,k < N} \big|^{\frac{1}{2}}.
	\end{eqnarray*}
	Note that this expression is invariant under any permutation of
	the points (this is also shown explicitly in Appendix A).
	We have then the following geometric generalization of the $N=2$ case
	in Theorem \ref{thm:many-particle-separation} 
	(or \ref{thm:many-particle-Hardy-3d}):

\begin{thm}[Volume of the $N-1$-simplex of $N$ points in $\R^d$, $d > N$ or $d = N-1$] 
	\label{thm:volume-noncritical}
	Consider $\Omega := \{ (\bx_1,\ldots,\bx_{N}) \in \R^{dN} : V(x) > 0 \}$
	and the ground state $f(x) := ( (N-1)! \,
		V(x) )^{-(d-N)}$. 
	Denoting 
	\begin{equation} \label{simplex-volume-ratio}
		\Sigma^{(N)}(x) := \frac{ 
			\sum_{k=1}^N V(\bx_1,\ldots,\check{\bx}_{k},\ldots,\bx_N)^2 
			}{ (N-1)^2 \, V(\bx_1,\ldots,\bx_N)^2 }
	\end{equation}
	(where $\check{}$ means deletion),
	we then have
	\begin{eqnarray} \label{simplex-Hardy}
		\int_\Omega |\nabla u|^2 \,dx
		- \frac{(d-N)^2}{4} \int_\Omega \Sigma^{(N)} |u|^2 \,dx
		= \int_\Omega |\nabla f^{-\frac{1}{2}} u|^2 f \,dx  \ \ge \ 0,
	\end{eqnarray}
	for all $u \in C_0^\infty(\Omega)$.
	The corresponding sharp Hardy inequalities hold on
	$H^1(\R^{dN})$ for $d>N$ and $H^1_0(\Omega)$ for $d=N-1$.
\end{thm}

	Before proving this theorem, 
	it is convenient to introduce the following notation:
	\begin{eqnarray*}
		A_k(x) &:=& (-1)^{k-1} \bigwedge_{1 \le j \neq k < N} (\bx_j - \bx_N), \quad k=1,\ldots,N-1, \\
		A_{k=N}(x) &:=& (-1)^{N-1} \bigwedge_{1 \le j < N-1} (\bx_j - \bx_{N-1}), \quad \textrm{and} \\
		A(x) &:=& \bigwedge_{1 \le j < N} (\bx_j - \bx_N),
	\end{eqnarray*}
	so that
	$$
		\Sigma^{(N)}(x) = \sum_{k=1}^N \left| A_k \liprod A^{-1} \right|^2
		= \sum_{k=1}^N |A_k|^2/|A|^2.
	$$
	Note that $\Sigma^{(N)}$
	describes a ratio of a mean of squares of volumes of 
	all $N-2$-dimensional
	subsimplices $A_k$ to the square of the volume of the full simplex $A$.
	In particular, for $N = 3$,
	\begin{equation} \label{volume-circumradius-relation}
		\Sigma^{(3)} 
		= \frac{r_{12}^2 + r_{23}^2 + r_{31}^2}{4V(\bx_1,\bx_2,\bx_3)^2}
		\ge \frac{4}{27} \frac{\rho^2}{(R_{123})^4},
	\end{equation}
	where we used that the simplex area $V$ 
	is bounded by $\frac{3\sqrt{3}}{4\pi}$ 
	times the area of the circumcircle. 

\begin{proof}[Proof of Theorem \ref{thm:volume-noncritical}]
	Note that we have defined $A_{k<N}$ and $A_N$ s.t.
	(cp. Appendix A)
	\begin{equation} \label{simplex-invariance}
		A = (\bx_k - \bx_N) \wedge A_k = (\bx_N - \bx_{N-1}) \wedge A_N.
	\end{equation}
	For each fixed $k < N$ we then have 
	$f(x) = |(\bx_k - \bx_N) \wedge A_k|^{-(d-N)}$
	and just as in Theorem \ref{thm:subspace-Hardy} that 
	$\nabla_{\bx_k} f = -(-1)^{\binom{N-1}{2}}(d-N) (A_k \liprod A^{-1}) f$
	and $\Delta_{\bx_k} f = 0$,
	and hence the following optimal single-particle Hardy GSR:
	\begin{equation} \label{simplex-GSR}
		\int_\Omega |\nabla_k u|^2 \,dx
		- \frac{(d-N)^2}{4} \int_\Omega |A_k \liprod A^{-1}|^2 |u|^2 \,dx
		= \int_\Omega |\nabla_k f^{-\frac{1}{2}} u|^2 f \,dx.
	\end{equation}
	For $k = N$ we use the invariance of $V$ under permutations,
	i.e. \eqref{simplex-invariance},
	and write instead
	$f(x) = |(\bx_N - \bx_{N-1}) \wedge A_N|^{-(d-N)}$,
	with analogous conclusions.
	Hence, $|\nabla f|^2 = \sum_k |\nabla_k f|^2 = (d-N)^2 \Sigma^{(N)} f^2$,
	$\Delta f = \sum_k \Delta_k f = 0$,
	and the GSR \eqref{simplex-Hardy} follows.
	Finally, note that in this case 
	$\codim \Omega^c = dN - (N-1)d - (N-2) = d-N+2$
	and hence $H_0^1(\Omega) = H^1(\R^{dN})$ for $d>N-1$.
	Sharpness is proven in Appendix B.
\end{proof}

	Again, for the codimension-critical case $d=N$ we fix a length
	scale $R>0$ and restrict to, e.g.,
	$\Omega_R := \{ (\bx_1,\ldots,\bx_{N}) \in \R^{d^2} : V(x) < R \}$.
	
\begin{thm}[Volume of the $d-1$-simplex of $d$ points in $\R^d$]
	\label{thm:volume-critical}
	Consider $\Omega := \{ (\bx_1,\ldots,\bx_d) \in \R^{d^2} : 0 < V(x) < R \}$
	and the ground state $f(x) := \left| \ln \frac{1}{R}V(x) \right|$. 
	Denoting 
	\begin{equation} \label{simplex-volume-ratio-critical}
		\tilde{\Sigma}^{(d)}(x) := \frac{ 
			\sum_{k=1}^d V(\bx_1,\ldots,\check{\bx}_{k},\ldots,\bx_d)^2 
			}{ (d-1)^2 \, V(\bx_1,\ldots,\bx_d)^2 \left( \ln \frac{1}{R}V(\bx_1,\ldots,\bx_d) \right)^2 },
	\end{equation}
	we then have
	\begin{eqnarray} \label{simplex-Hardy-critical}
		\int_\Omega |\nabla u|^2 \,dx
		- \frac{1}{4} \int_\Omega \tilde{\Sigma}^{(d)} |u|^2 \,dx
		= \int_\Omega |\nabla f^{-\frac{1}{2}} u|^2 f \,dx  \ \ge \ 0,
	\end{eqnarray}
	for all $u \in C_0^\infty(\Omega)$.
	The corresponding Hardy inequality holds on $H^1_0(\Omega_R)$.
\end{thm}
\begin{proof}
	This is completely analogous to the proof of 
	Theorem \ref{thm:volume-noncritical}, using that
	$\nabla_{\bx_k} f = -(-1)^{\binom{d-1}{2}} (A_k \liprod A^{-1})$
	and $\Delta_{\bx_k} f=0$.
\end{proof}

	Now, whereas Theorem \ref{thm:parallel-no3d} 
	involved all possible quantities
	$|B_{ij}| = |(\bx_i - 0) \wedge (\bx_j - 0)|$
	among $N$ points, 
	i.e. (twice) the volumes of all 2-simplices 
	spanned by selections of points of the form 
	$\{\bx_i,\bx_j,0\}$, this can naturally be generalized to
	higher dimensions as follows.
	This is both a geometrically and combinatorially
	more complete generalization of 
	Theorem \ref{thm:many-particle-Hardy-3d},
	which corresponds to $p=2$:

\begin{thm}[Volumes of all simplices of $p$ points 
	among $N$ points in $d \neq p$]
	\label{thm:volume-all}
	Consider $\Omega := \{ (\bx_1,\ldots,\bx_N) \in \R^{dN} : \mathbb{V}(x) > 0 \}$
	where
	$$
		\mathbb{V}(x) := \frac{1}{((p-1)!)^{\binom{N}{p}}} \prod_{1 \le j_1 < \ldots < j_{p} \le N} 
		\big| (\bx_{j_1} - \bx_{j_{p}}) \wedge \ldots \wedge (\bx_{j_{p-1}} - \bx_{j_{p}}) \big|
	$$
	is the product of the volumes of all $p-1$-simplices in $\R^d$ 
	spanned by $p$ of the points $\{\bx_{j=1,\ldots,N}\}$,
	and take the ground state  
	$f := \left( ((p-1)!)^{\binom{N}{p}} \mathbb{V} \right)^{-(d-p)}$.

	Denote by $\Lambda = \Lambda(p,N)$ the set of ordered subsets 
	$\lambda = (\lambda_1,\ldots,\lambda_p) \subseteq \{1,\ldots,N\}$
	of $p$ elements out of $N$.
	For $\lambda \in \Lambda$ define the $p-1$-blade
	$A_\lambda := A(\bx_{\lambda_1},\ldots,\bx_{\lambda_p})$,
	and for each $k \in \lambda$ let $A_{\lambda,k}$ denote a $p-2$-blade s.t.
	$A_\lambda = (\bx_k - \bx_{\lambda_q}) \wedge A_{\lambda,k}$
	for some $\lambda_q \in \lambda \setminus k$.
	With
	$$
		\Sigma_1^{(p,N)}(x) 
		:= \sum_{k=1}^N \sum_{\lambda \in \Lambda \,:\, k \in \lambda} \left| A_{\lambda,k} \liprod A_\lambda^{-1} \right|^2
		= \sum_{\lambda \in \Lambda} \sum_{k \in \lambda} |A_{\lambda,k}|^2/|A_\lambda|^2,
	$$
	and
	$$
		\Sigma_2^{(p,N)}(x) := \sum_{k=1}^N \sum_{ \substack{ \lambda,\mu \in \Lambda \,: \\ k \in \lambda \neq \mu \ni k} } 
			\left( A_{\lambda,k} \liprod A_\lambda^{-1} \right)
			\cdot \left( A_{\mu,k} \liprod A_\mu^{-1} \right),
	$$
	we then have
	\begin{eqnarray}
		\int_\Omega |\nabla u|^2 \,dx
		- (d-p)^2 \int_\Omega \left( 
			\alpha(1-\alpha) \Sigma_1^{(p,N)} 
			- \alpha^2 \Sigma_2^{(p,N)} 
			\right) |u|^2 \,dx
			\nonumber \\
		= \int_\Omega |\nabla f^{-\alpha} u|^2 f^{2\alpha} \,dx  \ \ge \ 0,
		\label{many-particle-volume-GSR}
	\end{eqnarray}
	for all $u \in C_0^\infty(\Omega)$.
	Hence, if 
	$C_{d,N}^{(p)} := \sup_{x \in \Omega} \Sigma_2^{(p,N)}(x) / \Sigma_1^{(p,N)}(x)$,
	then
	\begin{equation}
		\int_\Omega |\nabla u|^2 \,dx
		- \frac{(d-p)^2}{4(1 + C_{d,N}^{(p)})} \int_\Omega \Sigma_1^{(p,N)} |u|^2 \,dx
		\ge \int_\Omega |\nabla f^{-\alpha} u|^2 f^{2\alpha} \,dx  \ \ge \ 0,
		\label{many-particle-volume-1}
	\end{equation}
	and,
	\begin{equation}
		\int_\Omega |\nabla u|^2 \,dx
		- \frac{(d-p)^2}{4C_{d,N}^{(p)}(1 + C_{d,N}^{(p)})} \int_\Omega \Sigma_2^{(p,N)} |u|^2 \,dx
		\ge \int_\Omega |\nabla f^{-\alpha} u|^2 f^{2\alpha} \,dx  \ \ge \ 0,
		\label{many-particle-volume-2}
	\end{equation}
	with $\alpha := \frac{1}{2(1 + C_{d,N}^{(p)})}$.
\end{thm}
\noindent
	Note that by Cauchy-Schwarz in $\R^d$,
	\begin{eqnarray*}
		\Sigma_2^{(p,N)} 
		&\le& \sum_k \sum_{\lambda \ni k} \sum_{\substack{\mu \ni k \\ \mu \neq \lambda}}
		|A_{\lambda,k} \liprod A_\lambda^{-1}| |A_{\mu,k} \liprod A_\mu^{-1}| \\
		&\le& \sum_k \sum_{\lambda \ni k} \sum_{\substack{\mu \ni k \\ \mu \neq \lambda}}
		\frac{1}{2}\left( |A_{\lambda,k} \liprod A_\lambda^{-1}|^2 
			+ |A_{\mu,k} \liprod A_\mu^{-1}|^2 \right) \\
		&=& \sum_k \sum_{\lambda \ni k} 
			|A_{\lambda,k} \liprod A_\lambda^{-1}|^2
			\cdot \sum_{\substack{\mu \ni k \\ \mu \neq \lambda}} 1
			\ = \left( \binom{N-1}{p-1} - 1 \right) \Sigma_1^{(p,N)},
	\end{eqnarray*}
	hence
	$C_{d,N}^{(p)} \le \binom{N-1}{p-1} - 1$.
	Furthermore, as in the single-volume case
	one finds that $\codim \Omega^c = d - p + 2$.
	We therefore have the following bound for the Laplacian
	in terms of a potential which can be interpreted as
	a geometrically generalized many-particle interaction:

\begin{cor}
	For $d > p$ we have the generalized many-particle Hardy inequality
	\begin{equation} \label{general-Hardy}
		\int\limits_{\R^{dN}} \!\! |\nabla u|^2 \,dx
		\ge \frac{(d-p)^2}{4(1 + C_{d,N}^{(p)})} \int\limits_{\R^{dN}} \!
			\sum_{\lambda \in \Lambda(p,N)} \!\!\!
			\frac{ \sum_{k \in \lambda} 
				V(\bx_{\lambda_1},\ldots,\check{\bx}_{k},\ldots,\bx_{\lambda_p})^2 
				}{ (p-1)^2 \, V(\bx_{\lambda_1},\ldots,\bx_{\lambda_p})^2 } |u|^2 \,dx,
	\end{equation}
	for all $u \in H^1(\R^{dN})$.
	The inequality holds for $u \in H^1_0(\Omega)$ when $d=p-1$.
\end{cor}
\begin{proof}
	We have $f = \prod_{\lambda \in \Lambda(p,N)} |A_\lambda|^{-(d-p)}$
	with $\Delta_k |A_\lambda|^{-(d-p)} = 0$ on $\Omega$ $\forall k,\lambda$,
	and hence
	\begin{eqnarray*}
		\nabla_k f &=& \sum_{\lambda \in \Lambda(p,N) \,:\, k \in \lambda} 
			\nabla_k |A_\lambda|^{-(d-p)}
			\prod_{\mu \in \Lambda(p,N) \,:\, \mu \neq \lambda} |A_\mu|^{-(d-p)} \\
		&=& -\frac{(d-p)}{2} f 
			\sum_{\lambda \in \Lambda(p,N) \,:\, k \in \lambda} 
			|A_\lambda|^{-2} \nabla_k |A_\lambda|^2,
	\end{eqnarray*}
	and
	\begin{eqnarray*}
		\Delta_k f &=& 
			\sum_{\lambda \ni k} \nabla_k |A_\lambda|^{-(d-p)} 
			\cdot
			\sum_{\substack{\mu \ni k \,:\\ \mu \neq \lambda}} \nabla_k |A_\mu|^{-(d-p)} 
			\prod_{\nu \neq \lambda,\mu} |A_\nu|^{-(d-p)} \\
		&=& \frac{(d-p)^2}{4} f 
			\sum_{\lambda \ni k} \sum_{\substack{\mu \ni k \,:\\ \mu \neq \lambda}} 
			|A_\lambda|^{-2} \nabla_k |A_\lambda|^2 \cdot |A_\mu|^{-2} \nabla_k |A_\mu|^2.
	\end{eqnarray*}
	Then, again using \eqref{blade-gradient} from Appendix A,
	$$
		\nabla_k |A_\lambda|^2 
		= \nabla_k |(\bx_k - \bx_{\lambda_q}) \wedge A_{\lambda,k}|^2
		= (-1)^{\binom{p-1}{2}} 2A_{\lambda,k} \liprod A_\lambda,
	$$
	we therefore obtain
	$$
		|\nabla f|^2 = \sum_k |\nabla_k f|^2 
		= (d-p)^2 f^2 \left( \Sigma_1^{(p,N)} + \Sigma_2^{(p,N)} \right)
	$$
	and
	$$
		\Delta f = \sum_k \Delta_k f = (d-p)^2 f \Sigma_2^{(p,N)}.
	$$
	The GSR \eqref{many-particle-volume-GSR}
	and inequalities \eqref{many-particle-volume-1}-\eqref{many-particle-volume-2}
	now follow as in the earlier theorems.
\end{proof}

	For the optimal large-$N$ dependence of the constants in these
	many-particle inequalities, it becomes relevant to
	study the ratio of the geometric quantities w.r.t. the
	optimal asymptotic probability distribution $\rho$ of points in $\R^d$;
	$$
		C_d^{(p,q)} := \sup_{\rho \ge 0 \ :\, \int_{\R^d} d\rho(\bx) = 1}
		\frac{ \int_{\R^{qd}} \Xi^{(p,q)}(\bx_1,\ldots,\bx_{q}) \prod_{k=1}^{q} d\rho(\bx_k) 
			}{ \int_{\R^{pd}} \Sigma^{(p)}(\bx_1,\ldots,\bx_{p}) \prod_{k=1}^p d\rho(\bx_k) },
	$$
	where $q = p+1,\ldots,2p-1$,
	$\Sigma^{(p)}$ was defined in \eqref{simplex-volume-ratio}, 
	and
	$$
		\Xi^{(p,q)} := \sum_{ \substack{ 1 \in \lambda,\mu \in \Lambda(p,q) \\ 
				:\, |\lambda \cap \mu| = 2p-q } } 
			\sum_{\pi \in S_{q}}
			\left( A_{\pi(\lambda),\pi(1)} \liprod A_{\pi(\lambda)}^{-1} \right)
			\cdot \left( A_{\pi(\mu),\pi(1)} \liprod A_{\pi(\mu)}^{-1} \right)
	$$
	are higher-dimensional generalizations of a single circumradius 
	contribution $R_{123}^{-2}$.
	Related optimizations involving lower-dimensional 
	geometric quantities have been discussed in
	Remarks 2.2(iv) in \cite{H-O2-Laptev-Tidblom:08} 
	(see also \cite{Tidblom:thesis}),
	Section 3.1 in \cite{Frank-H-O-Laptev-Solovej},
	and also e.g. Eqn. (2) in \cite{Nam:10}.

	It is of course possible to generalize Theorem \ref{thm:volume-all} 
	even further and
	consider the volumes of \emph{all} simplices among $N$ points 
	(i.e. all simplex dimensions simultaneously),
	including the case of critical codimension.
	We will not state the corresponding theorem explicitly here.

\section{Conclusions} \label{sec:conclusions}

	We have studied operator inequalities for the Laplacian,
	i.e. uncertainty principles, involving many-particle interaction
	potentials of increasingly generalized geometric forms.
	Such interactions appear e.g. in membrane matrix models 
	\cite{Lundholm:thesis}
	and in higher-dimensional generalized 
	Calogero-Sutherland models such as those studied in 
	\cite{Calogero-Marchioro:73,Murthy-Bhaduri-Sen:96,Ghosh-Rao:98,Calogero:99},
	and the bounds provided here address the issue of whether the
	spectrum of the corresponding models is bounded from below independently
	of the choice of external potentials, i.e. whether the uncertainty
	principle is strong enough to prevent many-body collapse in these cases.
	Furthermore, the explicit ground state representations provided in these bounds will be of use in the spectral analysis of the associated operators.
	We have furthermore illustrated the novel use of techniques from 
	geometric algebra to conveniently facilitate these types of technically 
	and geometrically complicated computations.

\section*{Appendix A: Some applications of geometric algebra}

	The \emph{geometric algebra} over $\R^d$ 
	is the exterior algebra $\bigwedge \R^d$
	together with the left- and right-interior products 
	$(A,B) \mapsto A \liprod B$, $(A,B) \mapsto A \riprod B$
	and the associative Clifford product $(A,B) \mapsto AB$,
	which are all inherited from the usual Euclidean scalar product 
	$(\bx,\by) \mapsto \bx \cdot \by$.
	A $p$-\emph{blade} $A$ is an exterior product of $p$ vectors
	and is uniquely determined by its corresponding 
	$p$-dimensional subspace $\bar{A} \subseteq \R^d$,
	orientation, and magnitude $|A| := \sqrt{AA^\dagger}$,
	where $A^\dagger := (-1)^{\binom{p}{2}}A$ 
	is the \emph{reverse} of $A$.
	Note that if $A,B$ are blades then their
	exterior and interior products $A \wedge B$ and $A \liprod B$
	are blades as well, and that 
	$A \liprod B = AB$ if $\bar{A} \subseteq \bar{B}$
	(see e.g. Section 3 in \cite{Lundholm-Svensson:09}).

	If $A$ is a $p$-blade then we have for the gradient
	\begin{equation} \label{blade-gradient}
		\nabla_{\bx} |\bx \wedge A|^2 
		= 2(\bx \wedge A)A^\dagger
		= 2A^\dagger(A \wedge \bx)
		= (-1)^{\binom{p+1}{2}}2A \liprod (\bx \wedge A).
	\end{equation}
	One way to see this is to note that it is trivially true for
	$A = 0$, and that for $A \neq 0$ we can write any point 
	$\bx$ in $\R^d$ uniquely as
	$$
		\bx = \bx A A^{-1} = (\bx \liprod A)A^{-1} + (\bx \wedge A)A^{-1}
		= \bx_{\parallel} + \bx_{\perp},
	$$
	where $A^{-1} = A^\dagger/|A|^2$,
	$\bx_{\parallel} := P_{\bar{A}}\bx$ is the orthogonal
	projection on $\bar{A}$ 
	and $\bx_{\perp} := (1-P_{\bar{A}})\bx$ 
	(the so-called \emph{rejection} on $\bar{A}$;
	see e.g. Section 3.3 in \cite{Lundholm-Svensson:09}). 
	Hence,
	\begin{eqnarray*}
		\nabla_{\bx} |\bx \wedge A|^2 
		&=& \nabla_{\bx} |(\bx \wedge A)A^{-1}|^2 |A|^2 \\
		&=& \nabla_{\bx_{\parallel}} |\bx_{\perp}|^2 |A|^2
			+ \nabla_{\bx_{\perp}} |\bx_{\perp}|^2 |A|^2 \\
		&=& 0 + 2\bx_{\perp} |A|^2 
		\ = \ 2(\bx \wedge A)A^{\dagger}.
	\end{eqnarray*}
	The other equalities in \eqref{blade-gradient} follow
	by taking the reverse.

	It then also follows that
	$$
		\left| \nabla_{\bx} |\bx \wedge A|^2 \right|^2
		= 4 \left( (\bx \wedge A)A^\dagger \right) \left( (\bx \wedge A)A^\dagger \right)^\dagger
		= 4|A|^2|\bx \wedge A|^2,
	$$
	\begin{eqnarray*}
		\nabla_{\bx} |\bx \wedge A|^2 \cdot \nabla_{\by} |\by \wedge B|^2 
		&=& 4((\bx \wedge A)A^\dagger) \cdot ((\by \wedge B)B^\dagger) \\
		&=& 4(A^\dagger \liprod (A \wedge \bx)) \cdot (B^\dagger \liprod (B \wedge \by)) \\
		&=& (-1)^{\binom{p+1}{2} + \binom{q+1}{2}}
			4(A \liprod (\bx \wedge A)) \cdot (B \liprod (\by \wedge B)),
	\end{eqnarray*}
	for $A,B$ $p$- resp. $q$-blades,
	and, by e.g. Exercise 3.6 in \cite{Lundholm-Svensson:09},
	$$
		\Delta_{\bx} |\bx \wedge A|^2 
		= 2 \nabla_{\bx} (\bx \wedge A)A^\dagger
		= 2 \sum_{k=1}^d e_k(e_k \wedge A)A^\dagger
		= 2(d-p)|A|^2,
	$$
	as well as
	(outside the support of the corresponding distribution)
	\begin{eqnarray*}
		\Delta_{\bx} |\bx \wedge A|^{2\beta} 
		&=& \beta(\beta-1) |\bx \wedge A|^{2\beta-4} \left| \nabla_{\bx}|\bx \wedge A|^2 \right|^2 \\
		&& \quad	+\ \beta |\bx \wedge A|^{2\beta-2} \Delta_{\bx}|\bx \wedge A|^2 \\
		&=& 2\beta(2\beta + d-p-2) |A|^2 |\bx \wedge A|^{2\beta-2},
	\end{eqnarray*}
	which is zero for $\beta = -(d-p-2)/2$.
	Furthermore,
	(again on a suitable domain)
	$$
		\nabla_{\bx} \ln \frac{1}{R}|\bx \wedge A| 
		= \frac{1}{2}|\bx \wedge A|^{-2} \nabla_{\bx} |\bx \wedge A|^2
		= (\bx \wedge A)^{-1} A^\dagger,
	$$
	$$
		\Delta_{\bx} \ln \frac{1}{R}|\bx \wedge A| 
		= \frac{1}{2}\nabla_{\bx} \cdot \left( |\bx \wedge A|^{-2} \nabla_{\bx} |\bx \wedge A|^2 \right)
		= (d-p-2)|A|^2|\bx \wedge A|^{-2}.
	$$
	
	Lastly, we have the following explicit invariance of the simplex
	volume under permutations of the points:
	
\begin{prop} 
	Let
	$A := \bigwedge_{1 \le j < N} (\bx_j - \bx_N)$ 
	be the $N-1$-blade associated to the points
	$(\bx_1,\ldots,\bx_N)$. 
	Then $A$ is invariant, up to a sign, under any permutation of the points.
\end{prop}
\begin{proof}	
	$A$ is due to the total antisymmetry of the exterior product clearly 
	invariant under any permutation $\sigma$ of
	$(\bx_1,\ldots,\bx_{N-1})$, up to a sign ($\sgn \sigma$).
	Furthermore, for any $1 \le k < N$,
	\begin{eqnarray*}
		\lefteqn{ A = (\bx_1 - \bx_N) \wedge \ldots \wedge (\bx_k - \bx_N) \wedge \ldots \wedge (\bx_{N-1} - \bx_N) = }\\
		&=& \bigwedge_{j=1}^{k-1} (\bx_j - \bx_k + \bx_k - \bx_N) \wedge (\bx_k - \bx_N) \wedge \bigwedge_{j=k+1}^{N-1} (\bx_j - \bx_k + \bx_k - \bx_N) \\
		&=& \bigwedge_{j=1}^{k-1} (\bx_j - \bx_k) \wedge (\bx_k - \bx_N) \wedge \bigwedge_{j=k+1}^{N-1} (\bx_j - \bx_k) \\
		&=& (-1)^{N-k} \bigwedge_{\substack{j=1 \\ j \neq k}}^N (\bx_j - \bx_k),
	\end{eqnarray*}
	by multilinearity and antisymmetry.
	The proposition then follows by composition of permutations.
\end{proof}

\section*{Appendix B: Sharpness of derived constants}

	For completeness, we briefly note in this appendix how 
	the explicit ground states $f$ in the derived Hardy GSRs
	also can be used to determine sharpness of the constants
	in the corresponding Hardy inequalities.

\begin{lem} \label{lem:sharpness}
	Suppose that $\Omega$ and $f$ in Proposition \ref{prop:GSR} are
	such that $\overline{\Omega} = \R^n$ and
	\begin{itemize}
		\item[i)] $\int_{\R^n} f^{1-\delta} e^{-|x|} \, dx$ \quad
			is uniformly bounded for small $\delta > 0$,
		\item[ii)] $\int_{\R^n} \frac{|\nabla f|^2}{f^2}f^{1-\delta} e^{-|x|} \, dx$ \quad
			is finite for small $\delta > 0$, 
			but $\to \infty$ as $\delta \to 0$.
	\end{itemize}
	Then, for every $\epsilon > 0$ there exists 
	$u_\delta := f^{\frac{1}{2}} f^{-\frac{\delta}{2}} e^{-\frac{1}{2}|x|} \in H^1(\R^n)$,
	with $\delta > 0$, s.t.
	\begin{equation} \label{Hardy-sharpness}
		\int_{\R^n} |\nabla u_\delta|^2 \,dx
		\le \left(\frac{1}{4} + \epsilon\right) 
			\int_{\R^n} \frac{|\nabla f|^2}{f^2} |u_\delta|^2 \,dx.
	\end{equation}
\end{lem}
\noindent
	If $f(x) \to 0$ as $x \to \partial\Omega$,
	then we reverse the sign on $\delta$ above and note that
	$u_\delta \in H^1_0(\Omega)$ for small $\delta$.
\begin{proof}
	Using that $u_\delta \in C^1(\Omega \setminus \{0\})$
	and
	$$
		|\nabla u_\delta| 
		= \left| \frac{\nabla f}{2f} u_\delta 
			- \frac{\delta \nabla f}{2f} u_\delta
			- \frac{\nabla|x|}{2} u_\delta \right|
		\le \frac{1 + |\delta|}{2} \frac{|\nabla f|}{f} |u_\delta|
		+ \frac{1}{2} |u_\delta|
	$$
	on $\Omega \setminus \{0\}$,
	we find by \emph{i)} and \emph{ii)} that $u_\delta \in H^1(\R^n)$
	for sufficiently small $\pmp\delta > 0$, and that
	\begin{eqnarray*}
		\| \nabla u_\delta \|_{L^2} 
		&\le& \left( \frac{1}{2} + \frac{|\delta|}{2} \right) 
			\left\| \frac{\nabla f}{f} u_\delta \right\|_{L^2}
			+ \frac{1}{2} \|u_\delta\|_{L^2} \\
		&\le& \left( \frac{1}{2} + \frac{|\delta|}{2} 
			+ \frac{\sqrt{C}}{2} \left\| \frac{\nabla f}{f} u_\delta \right\|_{L^2}^{-1} \right) 
			\left\| \frac{\nabla f}{f} u_\delta \right\|_{L^2},
	\end{eqnarray*}
	where $C < \infty$ denotes a bound for \emph{i)}.
	The result follows by taking $\delta \to 0$. 
\end{proof}

	The conditions \emph{i)} and \emph{ii)} above typically hold
	when $f$ is of the form $f \sim \delta_\Omega^{-(k-2)}$ 
	near $\partial\Omega$, 
	where $\delta_\Omega$ is the distance to $\Omega^c$ 
	and $k$ the codimension of $\Omega^c$.
	Let us prove this explicitly for some of the Hardy GSRs considered
	in this paper. 

\begin{prop}
	The constant $(d-p-2)^2/4$ in \eqref{subspace-Hardy} is sharp.
\end{prop}
\begin{proof}
	As in Appendix A,
	we split the space $\R^d$ into variables $\bx_{\parallel} \in \bar{A}$,
	and $\bx_\perp$ orthogonal 
	to $\bar{A}$. Then,
	since $f(\bx) \propto |\bx_\perp|^{-(d-p-2)}$,
	\begin{eqnarray*}
		\int_{\R^d} f^{1-\delta} e^{-|\bx|} dx
		&\propto& \int_{\R^p} \int_{\R^{d-p}} |\bx_\perp|^{-(d-p-2)(1-\delta)} e^{-|\bx|} \,d\bx_\perp \,d\bx_\parallel \\
		&\le& |\mathbb{S}^{d-p-1}| \int_{\R^p} e^{-\frac{1}{2}|\bx_\parallel|} \,d\bx_\parallel 
			\int_{r=0}^\infty r^{-(d-p-2)(1-\delta)} e^{-\frac{r}{2}} r^{d-p-1} \,dr,
	\end{eqnarray*}
	which is uniformly bounded for $0 < |\delta| < \delta_0$, and
	\begin{eqnarray*}
		\int_{\R^d} \frac{|\nabla f|^2}{f^2} f^{1-\delta} e^{-|\bx|} dx
		&\propto& \int_{\R^p} \int_{\R^{d-p}} |\bx_\perp|^{-2-(d-p-2)(1-\delta)} e^{-|\bx|} \,d\bx_\perp \,d\bx_\parallel \\
		&\propto& \int_{r=0}^\infty 
			\int_{\R^p} e^{-\sqrt{|\bx_\parallel|^2 + r^2}} \,d\bx_\parallel 
			\, r^{-1+(d-p-2)\delta} \,dr,
	\end{eqnarray*}
	which is finite for $0 < \pmp\delta < \delta_0$, 
	but tends to infinity when $\delta \to 0$.
	Hence, by Lemma \ref{lem:sharpness}, the constant in the 
	GSR potential 
	$\frac{|\nabla f|^2}{4f^2} = \frac{(d-p-2)^2}{4} \frac{1}{|\bx_\perp|^2}$
	is sharp for $d \neq p+2$.
\end{proof}

\begin{prop} \label{prop:sharpness-mp-separation}
	The constant $N((N-1)d-2)^2/4$ in \eqref{many-particle-separation} 
	is sharp for $(N-1)d > 2$.
\end{prop}
\begin{proof}
	The set $\Omega^c$ in Theorem \ref{thm:many-particle-separation}
	is a linear subspace of $\R^{dN}$ which can be parameterized
	by, say, $\bx_N \in \R^d$.
	We then have the ground state
	\begin{eqnarray*}
		f(x) := \rho^{4\alpha} 
		= \left( \sum_{i<N} |\by_i|^2 + \sum_{i<j<N} |\by_i - \by_j|^2 \right)^{-\frac{(N-1)d-2}{2}},
	\end{eqnarray*}
	where we for fixed $\bx_N$ define $\by_i := \bx_i - \bx_N$.
	Hence, using that
	$$
		|y|^2 := \sum_{i<N} |\by_i|^2 \ \le \ \rho^2 \ 
		\le \ \sum_{i<N} |\by_i|^2 + \sum_{i<j<N} \left( |\by_i|^2 + |\by_j|^2 \right)
		\le C |y|^2
	$$
	(here and in the following, $C$ will denote some unspecified  
	positive constants), we find
	\begin{eqnarray*}
		\int_{\R^{dN}} f^{1-\delta} e^{-|x|} dx
		&=& \int_{\R^d} \int_{\R^{(N-1)d}} \rho^{-((N-1)d-2)(1-\delta)} e^{-|x|} \,dy \,d\bx_N \\
		&\le& C \int_{\R^d} \int_{\R^{(N-1)d}}
			e^{-\sqrt{|x'|^2 + |\bx_N|^2}} |y|^{-((N-1)d-2)(1-\delta)} \,dy \,d\bx_N
	\end{eqnarray*}
	(where $x = (x',\bx_N)$),
	which is uniformly bounded for $0 < \delta < \delta_0$, and
	\begin{multline*}
		\int_{\R^{dN}} \frac{|\nabla f|^2}{f^2} f^{1-\delta} e^{-|x|} dx \\
		\sim C \int_{\R^d} \int_{\R^{(N-1)d}}
			e^{-\sqrt{|x'|^2 + |\bx_N|^2}} |y|^{-2-((N-1)d-2)(1-\delta)} \,dy \,d\bx_N,
	\end{multline*}
	which is finite for $0 < \delta < \delta_0$, 
	but tends to infinity as $\delta \to 0$.
	Hence, the GSR constant in \eqref{many-particle-separation}
	is sharp by Lemma \ref{lem:sharpness}.
\end{proof}

\begin{prop}
	The constant $(d-3)^2/4$ in \eqref{2-particle-area-Hardy} 
	is sharp.
\end{prop}
\begin{proof}
	Here $f(x) = |\bx_1 \wedge \bx_2|^{-(d-3)}$ and
	$\Omega^c = \{ (\bx_1,\bx_2) \in \R^{2d} : \bx_1 \wedge \bx_2 = 0 \}$
	is a cone-like set which can be parameterized
	by $\bx_1 \in \R^d \setminus \{0\}$, $\bx_2 \in \R \bx_1$,
	and $\bx_1 = 0$, $\bx_2 \in \R^d$.
	Hence, for each fixed $\bx_1 \neq 0$ we split 
	the second variable
	into $x_{2\parallel}$ along, and 
	$\bx_{2\perp}$ orthogonal to, the line $\R\bx_1$,
	write $|\bx_1 \wedge \bx_2| = |\bx_1||\bx_{2\perp}|$,
	and deduce (with $\delta' := (d-3)\delta$)
	\begin{multline*}
		\int_{\R^{2d}} f^{1-\delta} e^{-|x|} dx \\
		= \int_{\R^d \setminus \{0\}} \int_{\R^{d-1}} 
			\int_{\R} e^{-|x|} dx_{2\parallel}\,
			|\bx_{2\perp}|^{-(d-3) + \delta'} d\bx_{2\perp}\,
			|\bx_1|^{-(d-3) + \delta'} \,d\bx_1,
	\end{multline*}
	and
	\begin{eqnarray*}
		\lefteqn{ \int_{\R^{2d}} \frac{|\nabla f|^2}{f^2} f^{1-\delta} e^{-|x|} dx }\\
		&\propto& \int_{\R^{2d}} \frac{|\bx_1|^2}{|\bx_1 \wedge \bx_2|^2} f^{1-\delta} e^{-|x|} dx
			+ \int_{\R^{2d}} \frac{|\bx_2|^2}{|\bx_1 \wedge \bx_2|^2} f^{1-\delta} e^{-|x|} dx \\
		&=& 2\int_{\R^d \setminus \{0\}} \int_{\R^{d-1}} 
			\int_{\R} e^{-|x|} dx_{2\parallel}\,
			|\bx_{2\perp}|^{-2-(d-3) + \delta'} d\bx_{2\perp}\,
			|\bx_1|^{-(d-3) + \delta'} \,d\bx_1,
	\end{eqnarray*}
	with similar conclusions as in our earlier examples.
\end{proof}

\begin{prop}
	The constant $(d-N)^2/4$ in \eqref{simplex-Hardy} is sharp.
\end{prop}
\begin{proof}
	Here $f(x) = |A|^{-(d-N)}$, with
	$A = \bigwedge_{j=1}^{N-1} (\bx_j - \bx_N)$, and
	$\Omega^c = \{ x \in \R^{dN} : A = 0 \}$
	can be parameterized
	by $\bx_N \in \R^d$ and,
	defining $\by_j := \bx_j - \bx_N$ for each fixed $\bx_N$, by 
	\begin{equation} \label{cone-parametrization}
		\by_1 \in \R^d \setminus \{0\}, \quad
		\by_2 \in \R^d \setminus \R\by_1, \quad
		\by_3 \in \R^d \setminus \overline{\by_1 \wedge \by_2}, \quad
		\ldots, \quad
		\by_{N-1} \in \overline{B},
	\end{equation}
	with $B := \by_1 \wedge \ldots \wedge \by_{N-2}$,
	\emph{plus}, $\by_1 = 0$, $\by_{j=2,\ldots,N-1} \in \R^d$, 
	and so forth.

	Hence, for each fixed $\bx_N \in \R^d$ 
	and $\by_{j=1,\ldots,N-2}$ in general position as 
	in \eqref{cone-parametrization} 
	we split the last variable $\by_{N-1}$
	into $\by_\parallel \in \overline{B}$ and 
	$\by_\perp$ orthogonal to $\overline{B}$,
	write $|A| = |B \wedge \by_{N-1}| = |B||\by_\perp|$,
	and deduce (with $\delta' := (d-N)\delta$)
	\begin{eqnarray*}
		\int_{\R^{dN}} f^{1-\delta} e^{-|x|} dx
		= \int_{\R^d} \int_{\R^d \setminus \{0\}} \ldots 
			\int_{\R^d \setminus \overline{\by_1 \wedge \ldots \wedge \by_{N-3}}} 
			\int_{\R^{d-N+2}} 
			\int_{\R^{N-2}} e^{-|x|} d\by_\parallel \\
		\cdot \ 
			|\by_\perp|^{-(d-N) + \delta'} d\by_\perp\,
			|B|^{-(d-N) + \delta'} \,d\by_{N-2} \ldots d\by_1 \,d\bx_N,
	\end{eqnarray*}
	and
	\begin{multline*}
		\int_{\R^{dN}} \frac{|\nabla f|^2}{f^2} f^{1-\delta} e^{-|x|} dx 
		\ \propto \ \int_{\R^{dN}} \Sigma^{(N)} f^{1-\delta} e^{-|x|} dx \\
		= N\int_{\R^d} \int_{\R^d \setminus \{0\}} \ldots 
			\int_{\R^d \setminus \overline{\by_1 \wedge \ldots \wedge \by_{N-3}}} 
			\int_{\R^{d-N+2}} 
			\int_{\R^{N-2}} e^{-|x|} d\by_\parallel \\
		\cdot \ 
			|\by_\perp|^{-2-(d-N) + \delta'} d\by_\perp\,
			|B|^{-(d-N) + \delta'} \,d\by_{N-2} \ldots d\by_1 \,d\bx_N,
	\end{multline*}
	with similar conclusions as in our earlier examples.
\end{proof}

\end{document}